\documentclass[10pt,journal]{IEEEtran}
\usepackage{times}
\usepackage{setspace}
\usepackage{url}
\spacing{1}
\usepackage[utf8]{inputenc}
\usepackage[T1]{fontenc}
\usepackage{graphicx}		% For \includegraphics
\usepackage{wrapfig}
\usepackage[format=plain,font=footnotesize,labelfont=bf,labelsep=period]{caption}
\usepackage{sidecap} % For sidecaptions for figures
\usepackage{subfig}
\usepackage[export]{adjustbox}
\usepackage[font=small]{caption}
\usepackage{float}
\usepackage{dblfloatfix}
\usepackage{svg}

\usepackage{amsmath} % assumes amsmath package installed
\usepackage{amssymb}  % assumes amsmath package installed
\usepackage{amsthm}
\usepackage{mathtools}
\usepackage[normalem]{ulem}
\usepackage{paralist}	% For enumerations with better titled numbers
\usepackage[space]{grffile} % Space in filenames
\usepackage{color}
\usepackage{bbm}
\usepackage{placeins} % Floatbarrier
\usepackage{array}
\usepackage{siunitx}

\newtheorem{theorem}{Theorem}

\newtheorem{lemma}{Lemma}
\newtheorem{definition}{Definition}
\newtheorem{remark}{Remark}
\newtheorem{assumption}{Assumption}

\newcommand{\R}{\mathbb{R}}

\definecolor{darkblue}{RGB}{0,0,102}
\definecolor{lightblue}{RGB}{77,77,148}

\definecolor{gold}{RGB}{234, 170, 0}
\definecolor{metallic_gold}{RGB}{139, 111, 78}

\renewcommand{\baselinestretch}{.985}
\allowdisplaybreaks

\newcommand{\lmat}{\begin{bmatrix}}
\newcommand{\rmat}{\end{bmatrix}}

\usepackage{mathrsfs}
\usepackage{graphics} % for pdf, bitmapped graphics files
\usepackage{amsmath} % assumes amsmath package installed
\usepackage{amssymb}  % assumes amsmath package installed
\usepackage{amsthm}
\usepackage{cite}
\usepackage{bm}
\usepackage{acronym}
\usepackage{paralist}
\usepackage{float}
\usepackage{color}
\usepackage{epstopdf}
\usepackage{multicol}
\usepackage{tikz}
\usepackage{hyperref}
\usepackage{graphicx}
\usepackage{mathtools}
\usepackage{soul}
\usepackage[bottom]{footmisc}

\usepackage{amsmath} % assumes amsmath package installed
\usepackage{amssymb}  % assumes amsmath package installed
\usepackage{graphicx}
\usepackage{epstopdf}
\usepackage{amsmath}
\usepackage{mathtools}
\usepackage{dsfont}
\usepackage{tikz}
\usepackage{siunitx}
\usepackage{xcolor}

\makeatletter
\hypersetup{colorlinks=true}
\AtBeginDocument{\@ifpackageloaded{hyperref}
  {\def\@linkcolor{blue}
  \def\@anchorcolor{red}
  \def\@citecolor{red}
  \def\@filecolor{red}
  \def\@urlcolor{black}
  \def\@menucolor{red}
  \def\@pagecolor{red}
\begingroup
  \@makeother\`%
  \@makeother\=%
  \edef\x{%
    \edef\noexpand\x{%
      \endgroup
      \noexpand\toks@{%
        \catcode 96=\noexpand\the\catcode`\noexpand\`\relax
        \catcode 61=\noexpand\the\catcode`\noexpand\=\relax
      }%
    }%
    \noexpand\x
  }%
\x
\@makeother\`
\@makeother\=
}{}}
\makeatother

\newcommand{\Zp}{\mathbb{N}_{0+}}
\newcommand{\Rp}{\mathbb{R}_{0+}}

\IEEEoverridecommandlockouts

\def\BibTeX{{\rm B\kern-.05em{\sc i\kern-.025em b}\kern-.08em
    T\kern-.1667em\lower.7ex\hbox{E}\kern-.125emX}}

\makeatletter
\newcommand{\fixed@sra}{$\vrule height 2\fontdimen22\textfont2 width 0pt\shortrightarrow$}
\newcommand{\shortarrow}[1]{%
  \mathrel{\text{\rotatebox[origin=c]{\numexpr#1*45}{\fixed@sra}}}
}
\makeatother

\begin{document}

\title{\Large {\bf Multi-rate Control Design under Input Constraints via Fixed-Time Barrier Functions}}

\author{Kunal Garg, \and Ryan K. Cosner, \and Ugo Rosolia,  \and Aaron D. Ames and Dimitra Panagou
\thanks{K. Garg and D. Panagou would like to acknowledge the support of the Air Force Office of Scientific Research under the award number FA9550-17-1-0284 and of the National Science Foundation under the award number 1942907. R. K. Cosner, U. Rosolia, and A. D. Ames would like to acknowledge the support of the National Science Foundation under the award number 1932091. K. Garg and D. Panagou are with the Department of Aerospace Engineering, University of Michigan, Ann Arbor, MI, USA; \texttt{\{kgarg, dpanagou\}@umich.edu}.
 R. K. Cosner, U. Rosolia, and A. D. Ames are with the AMBER
lab at the California Institute of Technology, Pasadena, CA, USA, e-mail: \texttt{\{rkcosner, urosolia, ames\}@caltech.edu}.
}

}
\maketitle

\begin{abstract}
In this paper, we introduce the notion of periodic safety, which requires that the system trajectories periodically visit a subset of a forward-invariant safe set, and utilize it in a multi-rate framework where a high-level planner generates a reference trajectory that is tracked by a low-level controller under input constraints. We introduce the notion of fixed-time barrier functions which is leveraged by the proposed low-level controller in a quadratic programming framework. Then, we design a model predictive control policy for high-level planning with a bound on the rate of change for the reference trajectory to guarantee that periodic safety is achieved. We demonstrate the effectiveness of the proposed strategy on a simulation example, where the proposed fixed-time stabilizing low-level controller shows successful satisfaction of control objectives, whereas an exponentially stabilizing low-level controller fails.
\end{abstract}

\section{Introduction}
% \textbf{TO DO: Statement on motivation for periodic stability:}
Constraints requiring the system trajectories to evolve in some \textit{safe} set at all times while visiting some goal set(s) are common in safety-critical applications. Constraints pertaining to the convergence of the trajectories to certain sets within a fixed time often appear in time-critical applications, e.g., when a task must be completed within a given time interval. 
% \textit{Spatiotemporal} specifications impose spatial (state) as well as temporal (time) constraints on the system trajectories. 
Most popular approaches on the control synthesis under such specifications include quadratic programming techniques, where the safety requirements are encoded via control barrier functions (CBFs) and convergence requirements via control Lyapunov functions (CLFs), see e.g. \cite{ames2017control,garg2019prescribedTAC}, or via one function that encodes both the safety and convergence requirements \cite{lindemann2019control,li2018formally}.

Quadratic program (QP)-based approaches have gained popularity for control synthesis \cite{li2018formally,ames2017control,garg2021FxTSDomain,garg2019prescribedTAC} in real-time, since QPs can be solved efficiently.
% Most of the prior work (except for a few studies, e.g.\cite{ames2017control,garg2019prescribedTAC}) addresses control design for safety along with convergence without explicitly considering control input constraints.
% Such constraints are considered in \cite{ames2017control}, where performance and safety objectives are represented using CLFs and CBFs, respectively, along with control input constraints in a QP. 
Most of the prior work, except \cite{li2018formally,lindemann2019control}, deals with asymptotic or exponential convergence of the system trajectories to the desired goal set.  Fixed-time stability (FxTS) \cite{polyakov2012nonlinear} is a stronger notion of stability, where the time of convergence does not depend on the initial conditions. To address the problem of {FxTS} in the presence of input constraints, new Lyapunov conditions are proposed in \cite{garg2021FxTSDomain}, characterizing a domain of attraction for FxTS under input constraints. 

As argued in the recent article \cite{cohen2020approximate}, \textit{myopic} control synthesis approaches relying solely on QPs are susceptible to infeasibility. To circumvent this issue, combining a high-level planner with a low-level controller has become a popular approach \cite{herbert2017fastrack,yin2019optimization, smith2019continuous,singh2017robust,rosolia2020multi}. The underlying idea in these strategies is to design low-level controllers to track a reference trajectory, which is computed by a high-level planner using a simplified model.
In~\cite{herbert2017fastrack} the authors presented the FaSTrack framework where the error bounds are computed using Hamilton-Jacobi (HJ) reachability analysis. 
% Afterwards, FaSTrack applies the high-level control action if the tracking error does not exceed the error bounds, otherwise, it applies the safe controller given by the reachability analysis. 
This framework has been extended in~\cite{yin2019optimization}, where the authors used Sum-Of-Squares (SOS) to compute the tracking error bounds. The constraint on the planner and the tracking error bounds may also be updated using an iterative procedure as suggested in~\cite{smith2019continuous}. 
% SOS programming is used also in~\cite{yin2019optimization} to compute the error bounds. 
% The authors parametrized the constraints of the high-level Model Predictive Controller (MPC) as a function of the tracking error and computed the controller which maximized the planner constraint set. 
A different approach that uses Model Predictive Controller (MPC) for high-level planning has been presented in~\cite{singh2017robust} where the tracking controller is designed using control contraction metrics.

In this work, we introduce the notion of periodic safety where the system trajectories are required to remain in a safe set for all times and visit a subset of this safe set periodically. Inspired from \cite{rosolia2020multi}, we use a multi-rate control framework where the low-level controller and the high-level planner operate at different frequencies. The high-level planner is used to generate a reference trajectory in the interior of a subset of the safe set, and the low-level controller to track this reference trajectory. The contribution of this paper is twofold. First, we combine the concepts of fixed-time stable Lyapunov functions \cite{garg2021FxTSDomain} and control barrier functions \cite{ames2017control} to define the notion of fixed-time barrier functions. We use it in a provably feasible QP, guaranteeing fixed-time convergence to a neighborhood of the reference trajectory from a region of attraction under input constraints. Second, we design the constraints of the MPC problem to consider this region of attraction of the low-level controller in the high-level planner. Compared to~\cite{rosolia2020multi}, we limit the rate of change for the planned trajectory so that the low-level controller is able to track the resulting reference trajectory within a predefined error bound. The limitation on rate change along with the tracking within the chosen error bound helps the system achieve periodic safety. Furthermore, we demonstrate that such constraints, which guarantee the correct operation of the low-level controller, do not jeopardize the feasibility of the MPC problem. Simulations demonstrate the fixed-time stabilizing low-level controller successfully satisfying the state constraints while an exponentially stabilizing controller \cite{ames2017control} fails. %Simulation results demonstrate the benefit of using a fixed-time stabilizing low-level controller via a scenario where an exponentially stabilizing controller \cite{ames2017control} fails to satisfy the state constraints, while the proposed controller successfully does so. 

% \subsection{Preliminaries}
\textit{Notation:} The Minkowski sum of two sets $\mathcal{X},\mathcal Y\subset \mathbb{R}^n$ is denoted as $\mathcal{X}_T\oplus\mathcal{Y}$, and the Pontryagin difference as $\mathcal{X}\ominus\mathcal{Y}$. The set of positive integers, non-negative integers and non-negative reals is denoted as $\mathbb N$, $\Zp$, and $\Rp$, respectively. 
% For any $i\in \mathbb N$ and $T>0$, the time interval $[(i-1)T, iT)$ is denoted as $\mathcal T_i$. 

\section{Problem formulation}\label{sec:problemFormulation}
We first introduce the problem under study and then present some related background material.
% and then we describe the problem formulation.
% \subsection{Problem formulation}

\noindent \textbf{System Model}: We consider nonlinear control affine system of the following form:
\begin{equation}\label{eq:sysModel}
    \dot x = f\big( x \big) + g\big( x \big) u,
\end{equation}
where $f:\mathbb R^{n_x}\rightarrow\mathbb R^{n_x}$ and $g:\mathbb R^{n_x}\rightarrow\mathbb R^{n_x\times n_u}$ are locally Lipschitz continuous functions with $f(0) = g(0) = 0$, $u \in \mathbb{R}^{n_u}$ is the input and $x \in \mathbb{R}^{n_x}$ is the system state. The control objective is to design a controller $u:\Rp\times \R^{n_x} \to \mathcal{U} \subset \R^{n_u} $ such that solutions to the closed loop system: 
\begin{equation} \label{eq:closed loop kx}
    \dot{x} = f_\textrm{cl}(t,x) \triangleq f(x) + g(x) u(t,x), \; t_0 = 0,
\end{equation}
satisfy the state constraints: %The control objective is to design an input $u$ such that the following state and input constraints are satisfied:
\begin{equation}\label{eq:lowLevelCnstr}
    % x(t) \in \mathcal{X}_T_c, 
    \hspace{-5pt}x(t)\in \mathcal X, \; \forall  t \in \Rp, \; x(iT) \in \mathcal{X}_T, \;  \forall i\in \Zp,
    %\hspace{-5pt}u(t) \in \mathcal{U}, \; x(iT) \in \mathcal{X}_T, \; x(t)\in \mathcal X, \; \forall  t \in \Rp, \; \forall i\in \Zp,
\end{equation}
% for all $ t \in \Rp$ and for all $i \in \Zp$,
where $\mathcal X\subset\mathbb R^{n_x}$ and $\mathcal X_T \subset \mathcal{X} \ominus \mathcal D$ with $\mathcal D = \{x\; |\; \|x\|\leq d\}$ for some $d>0$. We assume that the input constraint set is given as $\mathcal{U} = \{u\; |\; A_uu\leq b_u\}$ for some $A_u\in \mathbb R^{m\times n_u}, b_u\in \mathbb R^{m}$. The time constant $T$ is a user-specified parameter that defines the update frequency of the planned trajectory, as will be further clarified in Section \ref{sec:midLayer}. %specified by the user and, as it will be clear later on, it defines the frequency at which the controller updates the planned trajectory.
% \textcolor{blue}{The control design objective also includes driving the closed-loop trajectories to a set $\mathcal X_G = \{x\; |\; \|x\|_2\leq r_g\}\subset \mathcal X$ for some user-defined $r_g>0$.}
The control objectives as described in \eqref{eq:lowLevelCnstr} require safety of the system in terms of forward invariance of the set $\mathcal X$, and \textit{periodic} fixed-time stability of the set $\mathcal X_T$, which means that the system trajectories need to visit this set at each discrete time $iT$, $i\in \Zp$. To capture these objectives, we introduce the notion of periodic safety. 

\begin{definition}[\textbf{Periodic safety}]\label{def:period stab}
Given the sets $\mathcal{X}_T, \mathcal{X}\subset \mathbb R^{n_x}$, with $\mathcal X_T\subset \mathcal X$, and a time period $T>0$, the set $\mathcal X_T$ is said to be periodically safe w.r.t. the safe set $\mathcal X$ for the closed-loop system \eqref{eq:closed loop kx} if for all $x(0)\in \mathcal X_T$, the following holds 
\begin{align}
    x(iT) \in \mathcal X_T, \quad x(t)\in \mathcal X, \; \forall \ i\in \mathbb N, \; \forall t\geq 0. 
\end{align}
\end{definition}

\begin{figure}[!ht]
    \centering
        \includegraphics[width=0.5\columnwidth,clip]{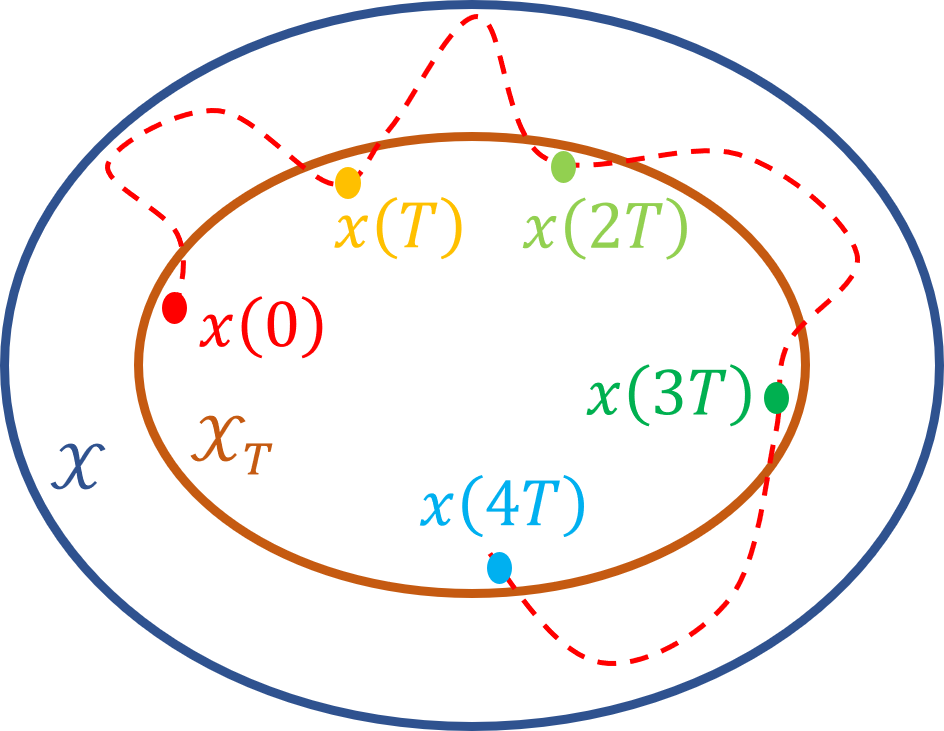}
    \caption{Illustration of periodic safety of the set $\mathcal X_T$ w.r.t. the set $\mathcal X$.}\label{fig:perStab1}
\end{figure}

\noindent Figure \ref{fig:perStab1} illustrates the periodic safety where the system trajectories visit the set $\mathcal X_T$ periodically, while remaining inside the safe set $\mathcal X$. Note that this notion is stronger than that of conditional invariance as defined in \cite{ladde1974flow}, where the set $\mathcal X$ is called conditionally flow-invariant for the closed-loop system~\eqref{eq:closed loop kx} if for all $x(0)\in \mathcal X_T\subset \mathcal X$, it holds that $x(t)\in \mathcal X$ for all $t\geq 0$. In particular, periodic safety of $\mathcal X_T$ w.r.t. $\mathcal X$ implies that $\mathcal X$ is conditionally flow-invariant.
% \subsection{Preliminaries}
% \noindent \textbf{Preliminaries:} 
Next, we define the notion of fixed-time domain of attraction: 
\begin{definition}\label{def:FxT-DoA}
Given a set $\mathcal{C}\subset \mathbb R^{n_x}$ and a time $T>0$, a set $D_{\mathcal{C}}\subset \mathbb R^{n_x}$ is a Fixed-Time Domain of Attraction (FxT-DoA) of the set $\mathcal{C}$ for the closed-loop system \eqref{eq:closed loop kx}, if 
\begin{itemize}
    \item[i)] for all $x(0) \in D_{\mathcal{C}}$, $x(t) \in D_{\mathcal{C}}$ for all $t\in [0, T)$, and
    \item[ii)] there exists $0\leq T_{\mathcal{C}}\leq T$ such that $\lim_{t\to T_{\mathcal{C}}}x(t) \in  \mathcal{C}$. 
\end{itemize}
\end{definition} 
\noindent The concept of FxT-DoA is important under a constrained input $u\in \mathcal U$, as it is not possible to guarantee that fixed-time convergence can be achieved for arbitrary initial conditions. To characterize this FxT-DoA, inspired from \cite{li2018formally}, we introduce a class of barrier functions termed fixed-time barrier functions. 
% and a set $S_{\mathcal D}$ such that $S\subset S_{\mathcal D}$. 

\begin{definition}\label{def:FxT barrier}
% A function $h:\mathbb R^{n_x}\rightarrow\mathbb R$ is called a FxT barrier for the the set $S\subset \mathbb R^{n_x}$ for system \eqref{eq:sysModel} with respect to the set $S_{\mathcal D} = S\oplus \mathcal D$ if
% \begin{itemize}
%     \item[i)] $h(x)<0$ for all $ x\notin S$ and  $h(x)>0$ for all $x\in \textnormal{int}(S)$;
%     \item[ii)] there exist $\delta, \alpha>0$, $\gamma_1>1$ and $0<\gamma_2<1$ such that the following holds for all $x\in S_{\mathcal D}$
%     \begin{align}\label{eq:FxTBarrier}
%     \sup_{u\in \mathcal U}\{L_fh(x) + L_gh(x)u\}\geq &-\delta h(x) + \alpha  \min\{0,h(x)\}^{\gamma_1}\nonumber \\
%     & +\alpha \min\{0,h(x)\}^{\gamma_2}.
% \end{align}
% \end{itemize}
A continuously differentiable function $h:\mathbb R^{n_x}\rightarrow\mathbb R$ is a FxT barrier function for the the set $\mathcal{S} = \{x\; |\; h(x)\geq 0\}$ with time $T_{\mathcal S}>0$ for the closed-loop system \eqref{eq:closed loop kx} if there exist parameters $\delta\in \mathbb R$, $\alpha>0$, $\gamma_1 = 1+\frac{1}{\mu}$ and $\gamma_2=1-\frac{1}{\mu}$ for some $\mu>1$, such that the following holds:
\begin{align}\label{eq:FxTBarrier}
\hspace{-10pt}\dot h(x)\geq -\delta h(x) + \alpha  \min\{0,h(x)\}^{\gamma_1}+\alpha \min\{0,h(x)\}^{\gamma_2},
\end{align}
for all $x\in D_{\mathcal S}\subset\mathbb R^{n_x}$ where $T_{\mathcal S}$ and $D_{\mathcal S}$ are functions of $\frac{\delta}{2\alpha}$.
% \begin{align*}
%     D_{\mathcal S} & = \begin{cases}
%     \mathbb R^{n_x}; &  r<1,\\
%     \left\{x\; |\; h(x)\geq -k^\mu\left(r-\sqrt{r^2-1}
%     % \frac{\delta -\sqrt{\delta ^2-4\alpha^2}}{2\alpha}
%     \right)^\mu\right\}; & r\geq 1,
%     %  \left\{x\; |\; V(x)\leq k^\mu\Big(\frac{\alpha }{\alpha }\Big)^\frac{\mu}{2}\right\}; & \delta  = 2\alpha,\\
%     \end{cases},
% \end{align*}
% where $r = \frac{\delta}{2\alpha}$ and $0<k<1$.
\end{definition}

\noindent Using \eqref{eq:FxTBarrier}, it follows from \cite[Theorem 1]{garg2021FxTSDomain} that the set $\mathcal D_{\mathcal S}$ is a FxT-DoA of the set $\mathcal S$ with time $T_{\mathcal S}$, where{\small
\begin{align*}
    D_{\mathcal S} & = \begin{cases}
    \mathbb R^{n_x}; &  r<1,\\
    \left\{x\; |\; h(x)\geq -k^\mu\left(r-\sqrt{r^2-1}
    % \frac{\delta -\sqrt{\delta ^2-4\alpha^2}}{2\alpha}
    \right)^\mu\right\}; & r\geq 1,
    %  \left\{x\; |\; V(x)\leq k^\mu\Big(\frac{\alpha }{\alpha }\Big)^\frac{\mu}{2}\right\}; & \delta  = 2\alpha,\\
    \end{cases},\\
    T_{\mathcal S} & = \begin{cases}
    \frac{\mu\pi}{\alpha\sqrt{1-\texttt{r}^2}}; &  \hspace{104pt}r<\texttt{r},\\
    \frac{\mu k}{\alpha(1-k)}; & \hspace{104pt} r\geq \texttt{r},
    \end{cases},
\end{align*}}\normalsize
with $r = \frac{\delta}{2\alpha}$ and $0<\texttt{r},k<1$. In particular, existence of a FxT barrier function $h$ implies: 1) forward invariance of the set $D_{\mathcal S}$ and 2) convergence to the set ${\mathcal S}$ within time $T_{\mathcal S}$.

\section{Multi-Rate Control}\label{sec:multi-rate strategy}
In this section, we present a hierarchical strategy where we first design a high-level planner that generates a reference trajectory $z(t)$, and then, a low-level controller that tracks this reference trajectory to guarantee that the closed-loop trajectory $x(t)$ satisfies \eqref{eq:lowLevelCnstr}. The control input is defined as
% \begin{equation}\label{eq:policy}
%     \pi(x) = \pi_l(x(t), z(t), u_m(t)) + \pi_m(x(t)) = u_l(t) + u_m(t),
% \end{equation} 
\begin{equation}
    u(t) = u_l(t) + u_m(t),
\end{equation} 
where $u_l$ and $u_m$ are defined using the policy $\Pi$ defined as\footnote{The closed-loop solutions of a sampled-data system are uniquely determined under piecewise continuous, bounded control inputs \cite[Section 2.2]{lars2011nonlinear}}: 
\begin{equation}\label{eq:policy}
    \begin{aligned}
    \Pi : \begin{cases} u_l(t) = \pi_l\big(x(t), u_m(t), i\big),~ \dot u_m(t) = 0, & t \in \mathcal{T}_i\\
     u_l^+(t) = u_l(t),~u_m^+(t) = \pi_m\big(x^+(t)\big), & t / T \in \mathbb{N} \end{cases},
\end{aligned}  
\end{equation}
where $\mathcal T_i = [(i-1)T, iT)$. Here, $\pi_m: \R^{n_x} \to \mathcal{U}_M\subset \mathcal U$ is the control policy from the high-level planner, to be designed in Section \ref{sec:midLayer}, that generates a reference trajectory using a Linear Time-Invariant (LTI) model of system \eqref{eq:sysModel}, and $\pi_l:\mathbb R^{n_x}\times\mathbb R^{n_u}\times\mathbb N\rightarrow\mathcal U$ is the low-level control policy, to be designed in Section \ref{sec: low level u}, that helps track this reference trajectory. The constraint set $\mathcal U_M\subset \mathcal U$ dictates how much of the control authority is reserved individually for the high-level planner and the low-level controller, and is a design parameter.

\subsection{High-level planning}\label{sec:midLayer}
In this section we describe the high-level planning strategy. \textbf{Reference Model:}  We assume that the reference trajectory $z(t)$ is generated using the following piecewise LTI model:
\begin{equation}\label{eq:referenceModel}
\begin{aligned}
    \Sigma_{z} : \begin{cases}
    \begin{matrix*}[l] \dot{z}(t) = A z(t)+ B u_m(t) \end{matrix*}, &  t \in \cup_{i=0}^\infty  (iT, (i+1)T)\\
    \begin{matrix*}[l] z^+(t) = \Delta_{z}(x^-(t)) \end{matrix*}, &  t \in \cup_{i=0}^\infty \{iT\} \\
    \end{cases},
\end{aligned}
\end{equation}
where $T$ from~\eqref{eq:lowLevelCnstr} is specified by the user and $z^-(t) = \lim_{\tau \nearrow t}z(\tau)$ and $z^+(t) = \lim_{\tau \searrow t}z(\tau)$ denote the right and left limits of the reference trajectory $z(t) \in \mathbb{R}^n$, which is assumed right continuous. The matrices $(A, B)$ are known and, in practice, may be computed by linearizing the system dynamics~\eqref{eq:sysModel} about the equilibrium point, i.e., the origin.
Finally, the reference input $u_m(t) \in \mathbb{R}^d$ and the \textit{reset map} $\Delta_{z}$, which depends on the state of the nonlinear system~\eqref{eq:sysModel}, are given by the higher layer as discussed next.

\vspace{2pt}
\noindent\textbf{Model Predictive Control:}
We design a Model Predictive Controller (MPC) to compute the high-level input $u_m(t)$ that defines the evolution of the reference trajectory in~\eqref{eq:referenceModel}, and to define the reset map $\Delta_{z}$ for the LTI model~\eqref{eq:referenceModel}. The MPC problem is solved at $1/T$ Hertz and therefore the reference high-level input is piecewise constant, i.e., $\dot u_m(t) =0~\forall t \in \mathcal{T}$ where $\mathcal T = \cup_{i=0}^\infty  (iT, (i+1)T)$. 
% \begin{equation*}
%     \dot u_m(t) =0~\forall t \in \mathcal{T} = \cup_{i=0}^\infty (iT, (i+1)T).
% \end{equation*}
First, we introduce the following discrete-time linear model:
\begin{equation}\label{eq:linearDiscreteSystem}
    z^d_{i+1} = \bar A z^d_i + \bar B v_i,
\end{equation}
where the transition matrices are $\bar A = e^{A T} \text{ and } \bar B = \int_0^{T} e^{A(T-\eta)}B d\eta$. Now notice that, as the high-level input  $u_m$ is piecewise constant, if at time $t_i = iT$ the state $z(iT)=z^+(iT)=z^d_i$ and $u_m(iT)=v_i$, then at time $t_{i+1} = (i+1)T$ we have that \begin{equation}\label{eq:relDisCon}
    z^-((i+1)T)=z^d_{i+1}.
\end{equation}
Given the discrete-time model~\eqref{eq:linearDiscreteSystem} and the state of the nonlinear system~\eqref{eq:sysModel} $x(iT)$, we solve the following finite-time optimal control problem at time $t_i = iT \in \mathcal{T}^c$:
% \begin{equation}
% \begin{aligned}
% J(x(iT)&)=\\
\begin{subequations}\label{eq:ftocp}
\begin{align}
    \min_{\boldsymbol{v}_i, z_{i|i}^d} \quad &\sum_{k = i}^{i+N-1}\big( || z_{k|i}^d||_Q + ||v_{k|i}||_R \big) + ||z_{i+N|i}^d||_{Q_f} \\
    \text{s.t.} ~\quad & z^d_{k+1|i} = \bar A z^d_{k|i} + \bar B v^d_{k|i}\\
    & ||z^d_{k+1|i} - z^d_{k|i}||_2 \leq d-c \label{eq: zi zi+1 close}\\ 
    &  z^d_{k|i} \in \mathcal{X}_T \ominus \mathcal{C}, ~ v^d_{k|i} \in \mathcal{U}_m \\
    &  z^d_{i|i} - x(iT) \in \mathcal{C} \label{eq: x z close}\\
    &  z^d_{i+N|i} \in \mathcal{X}_F ,\forall k = \{i, \ldots, i+N-1 \}
\end{align}
\end{subequations}
% \end{aligned}
% \end{equation}
where $||p||_Q = p^\top Qp$ and $\mathcal C = \{x\; |\; \|x\|\leq c\}$ for some $0<c<d$ such that $\mathcal X_T\ominus\mathcal C \neq \emptyset$. Problem~\eqref{eq:ftocp} computes a sequence of open-loop actions $\boldsymbol{v}_i^d=[v^d_{i|i},\ldots,v^d_{i+N|i}]$ and an initial condition $z^d_{i|i}$ such that the predicted trajectory steers the system to the terminal set $\mathcal{X}_F\subset \mathcal X_T$, while minimizing the cost and satisfying state and input constraints.
Let 
\begin{equation}\label{eq:mpcOpt}
    % \begin{aligned}
        \boldsymbol{v}_i^{d,*}=[v^{d,*}_{i|i},\ldots,v^{d,*}_{i+N|i}], \quad \boldsymbol{z}_i^{d,*}=[z^{d,*}_{i|i},\ldots,z^{d,*}_{i+N|i}]
    % \end{aligned}
\end{equation} be the optimal solution of \eqref{eq:ftocp}, then the high-level policy is \begin{equation}\label{eq:midLevPolicy}
\begin{aligned}
    \pi_{m}(x(iT)) = \begin{cases}
    \begin{matrix*}[l] {u_m}(t) = v^{d,*}_{i|i} \end{matrix*} &  t = iT \in \mathcal{T}^c\\
    \begin{matrix*}[l] \dot u_m(t) \!=\! 0 \end{matrix*} &  t \in \mathcal{T} \\
    \end{cases}
\end{aligned}
\end{equation}
Finally, we define the reset map for~\eqref{eq:referenceModel} as follows:
\begin{equation}\label{eq:returnMap}
\begin{aligned}
\Delta_{z}(x(iT)) = z_{i|i}^{d,*}.
\end{aligned}
\end{equation}
% The constraints in the MPC in \eqref{eq:ftocp} guarantee that the planned reference trajectory $z(t)$ satisfies $z^-(iT)\in \mathcal X_T\ominus\mathcal C$, and  $\|z^-(iT)-z^-((i+1)T)\| \leq d-c$ for all $i\in \Zp$. 

\subsection{Low-level control synthesis}\label{sec: low level u}
In this section we design the low-level policy $\pi_l$.
% so that the closed-loop trajectories satisfy $\|x(iT)-z^-(iT)\|\leq c$ for all $i\in \mathbb N$, which, along with the constraints imposed on $z(t)$, would imply that the closed-loop trajectories satisfy \eqref{ass:lowLevel}, as shown in Section \ref{sec:properties}. 
Consider the system dynamics \eqref{eq:sysModel} under the effect of the policy \eqref{eq:policy}:
\begin{align}\label{eq:closed_loop_system}
    \dot x(t) = f\big(x(t)\big) +g\big(x(t)\big)\big(u_l(t)+u_m(t)\big).
\end{align}
We define the sets $\mathcal D_i$ and $\mathcal C_i$ as
\begin{align}
   \mathcal D_{i} & \triangleq z^-(iT) \oplus \mathcal D = \{x\; |\; \|x-z^-(iT)\|\leq d\},\\
    \mathcal{C}_i & \triangleq z^-(iT) \oplus \mathcal C = \{x\; |\; \|x-z^-(iT)\|\leq c\}.
\end{align}
% Note that with $z^-(iT)\in \mathcal X_T\ominus \mathcal C$, and it follows that $\mathcal D_i\subset\mathcal X$ and $\mathcal C_i\subset \mathcal X_T$ for all $i\in \mathbb N$ (see Figure \ref{fig:PS temp}).  Furthermore, the constraint $\|z^-(iT)-z^-((i+1)T)\| \leq d-c$ helps guarantee that $\mathcal C_i\subset \mathcal D_{i+1}$. 
We show in Section \ref{sec:properties} that $\mathcal C_i\subset \mathcal D_{i+1}$ (guaranteed by bound on the rate change of the reference trajectory $z(t)$ in \eqref{eq: zi zi+1 close}) along with $\mathcal D_i\subset\mathcal X$ and $\mathcal C_i\subset \mathcal X_T$ (guaranteed by \eqref{eq: x z close}) guarantees that closed-loop trajectories meet the objectives in \eqref{eq:lowLevelCnstr}.
% \begin{figure}[!ht]
%     \centering
%         \includegraphics[width=0.5\columnwidth,clip]{figures/periodic_safety_illus_2.png}
%     \caption{Illustration of the sets $\mathcal C$ and $\mathcal D$ as the neighborhoods of the point $z_i = z^-(iT)$. 
%     % Here, the points $z_i$ are defined as $z_i = z^-(iT)$. 
%     }\label{fig:PS temp}
% \end{figure}
Under these considerations, the low-level control objective for $t\in \mathcal T_i = [(i-1)T, iT)$ is to design the policy $\pi_l$ such that the set $\mathcal D_i$ is FxT-DoA for the set $\mathcal{C}_i$.
% (see Figure \ref{fig:PS temp}).  
To this end, for the time interval $\mathcal T_i$ with $i\in \Zp$, consider the candidate FxT barrier function $h_i:\mathbb R^{n_x}\rightarrow\mathbb R$ defined as {
\begin{align}\label{eq: FxT barrier h_i}
    h_i(x(t)) = \frac{1}{2}c^2-\frac{1}{2}\|x(t)-z^-(iT)\|^2, \quad t\in \mathcal T_i.
\end{align}}
and define the following QP:{\small
\begin{subequations}\label{QP gen}
\begin{align}
\min_{u_l, \delta } \; &\frac{1}{2}u_l^2 + \frac{1}{2}  \delta ^2 + c\delta \\
    \textnormal{s.t.} \; &  A_u(u_l+u_m)  \leq  \; b_u, \label{C1 cont const}\\
    & L_fh_i(x) + L_gh_i(x)(u_m+u_l)  \geq  -\delta \hspace{1pt} h_i(x) \nonumber \\& \hspace{120pt}+\alpha \min\{0,h_i(x)\}^{\gamma_1} \nonumber\\
    &\hspace{120pt} +\alpha \min\{0,h_i(x)\}^{\gamma_2} \label{C2 stab const}
    % L_fh_S(x) + L_gh_S(x)v \leq & \; -\delta_2h_S(x),\label{C3 safe const}
\end{align}
\end{subequations}}\normalsize
where $c>0$ and $u_m = \pi_m(x(i-1)T)$. We denote the optimal solution of the QP \eqref{QP gen} as $(u_l^\star(x,u_m,i), \delta ^\star(x,u_m,i))$ and define the low-level policy as 
\begin{align}\label{eq:low_level_policy}
    \pi_l(x(t),u_m(t),i) = u_l^\star(x(t),u_m,i).
\end{align}
The constraint \eqref{C1 cont const} guarantees that $u = u_l+u_m\in \mathcal U$. The parameters $\mu, \alpha, \gamma_1, \gamma_2$ in \eqref{C2 stab const} are fixed, and are chosen as $\alpha  = \max\left\{\frac{\mu k}{(1-k)T},\frac{\mu \pi}{T\sqrt{1-\texttt{r}^2}}\right\}$, $\gamma_1 = 1+\frac{1}{\mu}$ and $\gamma_2 = 1-\frac{1}{\mu}$ with $\mu>1$ and $0<\texttt{r}, k<1$, so that the closed-loop trajectories reach the zero super-level set of the FxT barrier function $h_i$ within the time step $T$. 

\section{Closed-loop Properties}\label{sec:properties}
In this section we show the properties of the proposed multi-rate control architecture. Consider the closed-loop system \eqref{eq:closed_loop_system} under the control input \eqref{eq:policy} with policies $\pi_m$ and $\pi_l$ defined in \eqref{eq:midLevPolicy} and \eqref{eq:low_level_policy}, respectively. Below, we explain how we show that the closed-loop trajectories satisfy \eqref{eq:lowLevelCnstr}:
\begin{itemize}
    \item[A.] First, we show in Lemma \ref{lemma: DoA rM} that under the low-level controller $u_l$, the set $\mathcal D_i$ is FxT-DoA for the set $\mathcal{C}_i$, so that starting from any $x((i-1)T)\in \mathcal D_i$, the closed-loop trajectories reach the set $\mathcal{C}_i$ within time $T$. 
    \item[B.] Next, in Theorem \ref{th:main_result} we show recursive feasibility of the MPC so that the closed-loop trajectories satisfy $x((i-1)T)\in \mathcal D_i$ for all $i \in \mathbb N$, which along with item A, implies that the closed-loop trajectories satisfy \eqref{eq:lowLevelCnstr}. 
\end{itemize}
\subsection{Fixed Time Domain of Attraction}
% \noindent\textbf{Feasibility of the QP:} 
% \noindent \textbf{FxT-DoA of $\mathcal C_i$:} 
In this section, we show that under the low-level controller defined as the optimal solution of the QP \eqref{QP gen}, the set $\mathcal D_i$ is a FxT-DoA for the set $\mathcal{C}_i$. To this end, it is essential that the QP \eqref{QP gen} is feasible for all $x$ so that the low-level controller is well-defined. The slack term $\delta$ ensures the feasibility of the QP \eqref{QP gen} for all $x\notin \partial \mathcal{C}_{i}$. For the feasibility of the QP \eqref{QP gen} for $x\in \partial \mathcal{C}_{i}$, we make the following assumption, which is a standard assumption in the literature for guaranteeing forward invariance (see \cite{blanchini1999set} for more details). 

\begin{assumption}\label{assum: QP feas bd}
For all $x\in \partial \mathcal{C}_i$, $i\in \mathbb Z_+$, and $u_m\in \mathcal U_M$, there exists $u_l\in \mathcal U_l$ such that the following holds:
\begin{align*}
    L_fh_i(x) + L_gh_i(x)(u_m+u_l) \geq 0.
\end{align*}
\end{assumption}

From Definition \ref{def:FxT barrier}, we know that FxT-DoA depends on the ratio $\frac{\delta}{2\alpha}$. We make the following assumption on the maximum value of $\delta^\star(x)$ as the solution of the QP \eqref{QP gen} so that $h_i$ is a FxT barrier function for $\mathcal{C}_i$ and $\mathcal D_i$ is its FxT-DoA.

\begin{assumption}\label{ass:lowLevel}
% For each \textcolor{orange}{$i \in \Zp$}, it holds that  $\|z^-((i+1)T)-z^-(iT)\|\leq d-c$. Furthermore, 
For each interval $\mathcal T_i$, the solution $\left(u^\star(x(t),u_m,i), \delta ^\star(x(t),u_m,i)\right)$ of the QP \eqref{QP gen} is continuous for all $t\in \mathcal T_i$ and the following holds
% \begin{align}
%     r_M \leq 
% \end{align}
\begin{align}\label{eq:bar r}
    \sup_{t\in \mathcal T_i}\frac{\delta ^\star(x(t),u_m,i)}{2\alpha} \leq \bar r \triangleq \frac{(\frac{d^2-c^2}{2})^\frac{1}{\mu}}{2k} + \frac{k}{2(\frac{d^2-c^2}{2})^\frac{1}{\mu}},
    % \delta_M \leq \alpha\frac{\frac{r^\frac{2}{\mu}}{k^2}+2}{\frac{r^\frac{1}{\mu}}{k}},
\end{align}
where $\alpha = \max\{\frac{\mu k}{(1-k)T},\frac{\mu \pi}{T}\}$, $\gamma_1 = 1+\frac{1}{\mu}$ and $\gamma_2 = 1-\frac{1}{\mu}$ for some $\mu>1$ and $0<k<1$.
\end{assumption}

\begin{remark}
% Why this assumption makes sense? How can we check? Just run a bunch of sim? Maybe we can have something here saying why this is reasonable
As argued in \cite{garg2021FxTSDomain}, for given input bounds (dictated by the set $\mathcal U)$, the value of the slack term $\delta$ in QP \eqref{QP gen} depends on the time of convergence $T$. Furthermore, the upper-bound in \eqref{eq:bar r} depends on the parameters $c$ and $k$, where $0<c<d$ is such that $\mathcal X_T\ominus\mathcal C$ is non-empty and $0<k<1$. Thus, in practice, numerical simulations can guide the choice of the parameters $c,k$, and the time $T$, so that \eqref{eq:bar r} can be satisfied. 
\end{remark}

% Recall that $(u_l^\star(x), \delta^\star(x))$ denotes the optimal solution of the QP \eqref{QP gen}. 
% \begin{lemma}\label{lemma: DoA rM}
% Consider the QP \eqref{QP gen} for the time interval $\mathcal T_i = [iT, (i+1)T)$. Let $r_M = \sup_{t\in \mathcal T_i} r^\star(x(t))$ with $r^\star(x(t)) \coloneqq \frac{\delta ^\star(x(t))}{2\alpha}$. Then, the domain of attraction $D$ for the domain $\mathcal{C}_{i+1}$  is given as 
% \begin{align}\label{eq: DoA express rM}
%     D & =  \begin{cases}\mathbb R^{n_x}; & r_M<1,\\
%     \left\{x\; |\; V(x)\leq k^\mu\left(r_M-\sqrt{r_M^2-1}\right)^\mu\right\}; & r_M\geq 1,
%     \end{cases}.
% \end{align}
% \end{lemma}

\begin{lemma}\label{lemma: DoA rM}
Suppose that Assumptions \ref{assum: QP feas bd}-\ref{ass:lowLevel} hold. Then, for each time interval $\mathcal T_i$ with $i \in \mathbb N$, under the control policy \eqref{eq:low_level_policy}, it holds that for all $x((i-1)T)\in \mathcal D_{i}$, the closed-loop trajectory $x(t)$ satisfies $x(t)\in \mathcal D_i$ for all $t\in \mathcal T_i$ and $x^-(iT)\in \mathcal{C}_{i}$.
\end{lemma}

\begin{proof}
Under Assumption \ref{assum: QP feas bd}, it follows from \cite[Lemma 2]{garg2019prescribedTAC} that the QP \eqref{QP gen} is feasible for all $x$.
% In the particular case when $\alpha  = \alpha $, let us examine how the domain of attraction $D$ in \eqref{eq: domain of attraction} is affected by the ratio $r^\star$. Define $r_M = \sup_{x\in S}r^\star(x)$. 
Denote $D_{\mathcal{C}_i}$ as a FxT-DoA for the set $\mathcal{C}_i$ for the time $T$. Note that by definition, $\mathcal{C}_i = \{x\; |\; h_i(x)\geq 0\}$. We first compute an expression for $D_{\mathcal{C}_i}$ and then, we show that under Assumption \ref{ass:lowLevel}, $\mathcal D_i \subseteq D_{\mathcal{C}_i}$. From \cite[Theorem 1]{garg2021FxTSDomain}, we know that the FxT-DoA $D_{\mathcal C_i}$ is given as a function of $r_M = \sup r^\star = \frac{\delta^\star}{2\alpha}$, i.e., the maximum value  of the ratio $r^\star$. We consider the two cases, namely $r_M<1$ and $r_M\geq 1$ separately. 

For $r_M<1$, it follows from \cite[Theorem 1]{garg2021FxTSDomain} that $D_{\mathcal{C}_i} =\R^{n_x}$ is the FxT-DoA for $\mathcal{C}_i$. Thus, $\mathcal D_i$ is also a FxT-DoA of the set $\mathcal{C}_i$. 

For $r_M\geq 1$, it follows from \cite[Theorem 1]{garg2021FxTSDomain} that a FxT-DoA (i.e., the set $D_{\mathcal{C}_i}$) is given as{\small
\begin{align*}
   D_{\mathcal{C}_i} = \left\{x\; |\; h_i(x)\geq -\inf_{t\in \mathcal T_i}k^\mu\left(r^\star(x(t))-\sqrt{(r^\star(x(t)))^2-1}\right)^\mu\right\}. 
\end{align*}}\normalsize
Note that{
% \begin{align}
%     D = \bigg\{x\; |\; & h_i(x)\geq -\inf_{t\in \mathcal T_i} k^\mu\left(r^\star(x(t))-\sqrt{(r^\star(x(t)))^2-1}\right)^\mu\bigg\} \nonumber \\
%     & =  \left\{x\; |\; h_i(x)\geq - k^\mu\left(r_M-\sqrt{r_M^2-1}\right)^\mu\right\} \label{eq:fxt_r_geq1} 
% \end{align}
\begin{align}
    \hspace{-10pt} \inf_{t\in \mathcal T_i} r^\star(x(t))-\sqrt{(r^\star(x(t)))^2-1} = r_M-\sqrt{r_M^2-1} \label{eq:fxt_r_geq1} 
\end{align}}\normalsize
where the equality follows from the fact that $\left(r-\sqrt{r^2-1}\right)$ is a monotonically decreasing function for $r\geq 1$. Thus, it follows that
$D_{\mathcal{C}_i} =  \left\{x\; |\; h_i(x)\geq - k^\mu\left(r_M-\sqrt{r_M^2-1}\right)^\mu\right\}$. With the FxT barrier function $h_i(x)$ defined in \eqref{eq: FxT barrier h_i}, the set $D_{\mathcal{C}_i}$ reads{\small
\begin{align}\label{eq: d_e set rm geq 1}
    D_{\mathcal{C}_i}
    % & = \left\{x\; |\; \frac{1}{2}c^2-\frac{1}{2}\|x-z^-(iT)\|^2\geq- k^\mu\left(r_M-\sqrt{r_M^2-1}\right)^\mu\right\} \nonumber\\
    & = \left\{x\; |\; \frac{1}{2}\|x-z^-(iT)\|^2\leq k^\mu\left(r_M-\sqrt{r_M^2-1}\right)^\mu+\frac{1}{2}c^2.\right\}
\end{align}}\normalsize
Now, under Assumption \ref{ass:lowLevel}, it holds that $r_M \leq  \frac{(\frac{d^2-c^2}{2})^\frac{1}{\mu}}{2k} + \frac{k}{2(\frac{d^2-c^2}{2})^\frac{1}{\mu}}$.
By re-arranging this inequality, we obtain that under Assumption \ref{ass:lowLevel}, it holds that
% in order for the set $\mathcal{D}_{i}$ to be the FxT-DoA for the set $\mathcal{C}_i$, in time $T$ as desired, it is required that $\mathcal{D}_{i} \subseteq D$, or, equivalently, it is required that 
\begin{align}\label{eq: r_m r c rel}
    \frac{1}{2}d^2 \leq k^\mu\left(r_M-\sqrt{r_M^2-1}\right)^\mu + \frac{1}{2}c^2.
\end{align}
Now, for any $x((i-1)T)\in \mathcal D_i$, it holds that $\|x((i-1)T)-z^-(iT)\|\leq d$. Thus, it follows from \eqref{eq: r_m r c rel} that $$\frac{1}{2}\|x((i-1)T)-z^-(iT)\|^2\leq k^\mu\left(r_M-\sqrt{r_M^2-1}\right)^\mu + \frac{1}{2}c^2,$$ for all $x((i-1)T)\in \mathcal D_i$. Using this, and \eqref{eq: d_e set rm geq 1}, it follows that $\mathcal D_i = D_{\mathcal{C}_i}$. Hence, we have that $\mathcal D_i$ is a FxT-DoA of the set $\mathcal{C}_i$. Thus, from \cite[Theorem 1]{garg2021FxTSDomain}, it follows that the closed-loop trajectories of \eqref{eq:closed loop kx} will reach the set $\mathcal{C}_{i}$ for any $x((i-1)T)\in \mathcal D_{i}$ within a fixed time $T_1$ that satisfies $T_1\leq \max\{\frac{\mu k}{\alpha(1-k)}, \frac{\mu\pi}{\alpha\sqrt{1-\texttt{r}^2}}\}$. For the choice of $\alpha = \max\{\frac{\mu k}{(1-k)T},\frac{\mu \pi}{T\sqrt{1-\texttt{r}^2}}\}$, it follows that $T_1\leq T$. Thus, the system trajectories reach the set $\mathcal{C}_{i}$ on or before $t = (i-1)T + T = iT$. 

Finally, we show that the closed-loop trajectories remain in the set $\mathcal{C}_i$ till $t = iT$, i.e., the set $\mathcal C_i$ is forward invariant for the closed-loop trajectories of \eqref{eq:closed loop kx}. Let $t = t_{i} \triangleq (i-1)T + T_1$ denote the first time instant when the closed-loop trajectories of \eqref{eq:closed loop kx} reach the boundary of the set $\mathcal{C}_{i}$, i.e., $h_i(x(t_i)) = 0$. From the analysis in the first part of the lemma, it holds that $t_{i}\leq iT$. From \eqref{C2 stab const}, it follows that $\dot h_i(x)\geq -\delta^\star(x,u_m,i) h_i(x)\geq -\delta_M h_i(x)$ for all $x\in \mathcal{C}_{i}$, where $\delta_M = \sup_{t\in \mathcal T_i}\delta^\star(x(t),u_m,i)$. The proof can be completed using \cite[Proposition 1]{ames2017control}. 
\end{proof}

Thus, satisfaction of \eqref{C2 stab const} implies that system trajectories reach the set $\mathcal{C}_{i}$ on or before $t = iT$, and stay there till $t = iT$. Now, in order for the closed-loop trajectories to reach the set $\mathcal{C}_{i+1}$ on or before $t = (i+1)T$, it is required that $x(iT)\in \mathcal D_{i+1}$, which is shown in the following lemma.   

\begin{lemma}\label{prop:mpcFeasibilityIMplications}
If the MPC problem~\eqref{eq:ftocp} is feasible at time $t_i=iT$, then $x(iT) \in \mathcal D_{i+1}$.
\end{lemma}
\begin{proof}
By assumption of the lemma, the MPC problem \eqref{eq:ftocp} is feasible at time $t_i = iT$. Now consider the optimal MPC solution~\eqref{eq:mpcOpt} at time $t_i = iT$. By definition, we have that $x(iT) - z_{i|i}^{*,d} \in \mathcal{C}$, which implies that $||x(iT) - z_{i|i}^{*,d}|| \leq c$. Furthermore, by feasibility of the optimal MPC solution~\eqref{eq:mpcOpt} for problem~\eqref{eq:ftocp}, we have that 
$||z_{i+1|i}^{*,d}-z_{i|i}^{*,d}|| \leq d-c$. This implies that 
\begin{equation*}
\begin{aligned}
    ||x(iT) - z_{i+1|i}^{*,d}|| &= ||x(iT) -z_{i|i}^{*,d}+z_{i|i}^{*,d}- z_{i+1|i}^{*,d}|| \\
    &\leq ||x(iT) -z_{i|i}^{*,d}||+||z_{i|i}^{*,d}- z_{i+1|i}^{*,d}||\leq d
\end{aligned}
\end{equation*}
Finally, from~\eqref{eq:relDisCon} we have that $z_{i+1|i}^{*,d} = z^-((i+1)T)$. Thus, from the above equation we conclude that $||x(iT) - z^-((i+1)T)|| \leq d$, which implies that $x(iT) \in \mathcal D_{i+1}$.
\end{proof}

\subsection{MPC Recursive Feasibility and Closed-Loop Constraint Satisfaction}
So far, we have shown that feasibility of the MPC guarantees that $x(iT)\in \mathcal D_{i+1}$, which, under the low-level control policy \eqref{eq:low_level_policy}, guarantees that $x((i+1)T)\in \mathcal C_{i+1}$. Thus, what is remaining to be shown is that the MPC \eqref{eq:ftocp} is recursively feasible, i.e., if \eqref{eq:ftocp} is feasible at $t = 0T$, then it is feasible at $t = iT$ for all $i\in \mathbb N$. This would guarantee that $x(iT)\in \mathcal C_i$ (and hence, $x(iT)\in \mathcal X_T$) for all $i\in \mathbb N$.

\noindent \textbf{Recursive feasibility of MPC:} We make the following assumption for the high-level planner that would help guarantee recursive feasibility of the MPC \eqref{eq:ftocp}.

\begin{assumption}\label{ass:invariance}
The set $\mathcal{X}_F$ is invariant for the autonomous discrete time model $z^d((i+1)T) =\bar A z^d(iT)$ for all $i\in \mathbb{N}$. Furthermore, for all $i\in \mathbb N$, it holds that $||z^d(iT)-\bar Az^d(iT)|| \leq d - c$.
\end{assumption}

\begin{remark}
The above assumption is standard in the MPC literature~\cite{borrelli2017predictive, kouvaritakis2016model} and it allows us to guarantee that the MPC problem is feasible at all time instances. In practice, the set $\mathcal{X}_F$ can be chosen as a small neighborhood of the origin.
% and the region of attraction of the MPC policy are then defined over a superset of $\mathcal{X}_F$, which is given by all the states that can be steered to the set $\mathcal{X}_F$ while satisfying the state and the input constraints.
\end{remark}

Now we are ready to state the main result that shows that the hierarchical control strategy in Section \ref{sec:multi-rate strategy} leads to satisfaction of \eqref{eq:lowLevelCnstr}.
\begin{theorem}\label{th:main_result}
Let Assumptions~\ref{assum: QP feas bd}-\ref{ass:invariance} hold and consider the closed-loop system~\eqref{eq:closed_loop_system} under the control policy~\eqref{eq:policy}, where $\pi_m$ is defined in \eqref{eq:midLevPolicy} and $\pi_l$ is defined in \eqref{eq:low_level_policy}. If at time $t=0$ problem~\eqref{eq:ftocp} is feasible, then 
% , i.e., 
the closed-loop trajectories under the control policy \eqref{eq:policy} satisfy \eqref{eq:lowLevelCnstr}, i.e., the set $\mathcal X_T$ is periodically safe w.r.t. the set $\mathcal X$ with period $T$.
% state and input constraints~\eqref{eq:lowLevelCnstr} are satisfied for all time $t \in \Rp$.
\end{theorem}

\begin{proof}
The proof proceeds by induction.
First, we show that if at time $t_i = iT$ the MPC problem~\eqref{eq:ftocp} is feasible, then at time $t_{i+1}=(i+1)T$ the MPC problem~\eqref{eq:ftocp} is feasible. Let 
\begin{equation*}
    [z_{i|i}^{d,*}, z_{i+1|i}^{d,*},\ldots,z_{i+N|i}^{d,*}] \text{ and } [u_{i|i}^{d,*},\ldots,u_{i+N-1|i}^{d,*}]
\end{equation*}
be the optimal state input sequence to the MPC problem~\eqref{eq:ftocp} at time $t_i = iT$. Then from the feasibility of the MPC problem and Proposition~\ref{prop:mpcFeasibilityIMplications} we have that $x(iT) \in \mathcal D_{i+1}$, which together with Lemma \ref{lemma: DoA rM} implies that
\begin{equation*}
    x((i+1)T) \in \mathcal{C}_{i+1} =\{x ~| ~ || x-z^-((i+1)T)||\leq c\}.
\end{equation*}
Now notice that from equation~\eqref{eq:relDisCon} we have $z_{i+1|i}^{*,d} = z^-((i+1)T)$, which in turn implies that 
\begin{equation}\label{eq:discContConn}
    x((i+1)T)- z_{i+1|i}^{d,*} = x((i+1)T)- z^-((i+1)T) \in \mathcal{C}
\end{equation}
and therefore, by Assumption~\ref{ass:invariance}, the following sequences of states and inputs
\begin{equation}\label{eq:feasTr}
   \hspace{-5pt} [z_{i+1|i}^{d,*},\ldots,z_{i+N|i}^{d,*}, \bar A z_{i+N|i}^{d,*}],\;  [u_{i+1|i}^{d,*},\ldots,u_{i+N-1|i}^{d,*}, 0]
\end{equation}
are feasible at time $t_{i+1}=(i+1)T$ for the MPC problem~\eqref{eq:ftocp}. We have shown that if the MPC problem~\eqref{eq:ftocp} is feasible at time $t_i = iT$, then the MPC problem is feasible at time $t_{i+1}=(i+1)T$. Per assumption of the theorem, problem~\eqref{eq:ftocp} is feasible at time $t_0 =0$, and hence, we conclude by induction that the MPC problem~\eqref{eq:ftocp} is feasible for all $t_i = iT$ and for all $i \in \Zp$.

Next, we show that the feasibility of the MPC problem implies that state and input constraints are satisfied for the closed-loop system. Notice that by definition $u_m(t) = v_{i|i}^* \in \mathcal{U}_m $ for all $t \in [iT, (i+1)T)$ and from Lemma \ref{lemma: DoA rM}, we have the low-level controller returns a feasible control action $u_l(t)$, therefore we have that 
\begin{equation}
    u(t) = u_l(t) + u_m(t) \in \mathcal{U}, \forall t \in \Rp.
\end{equation}
Finally, from the feasibility of the state-input sequences in~\eqref{eq:feasTr} for the MPC problem~\eqref{eq:ftocp}, we have that 
\begin{equation}\label{eq:feasSolCnstr}
    \begin{aligned}
    & x_{i|i}^{d,*} \in \mathcal{X}_T \ominus \mathcal{C} \text{ and } x(iT) - x_{i|i}^{d,*} \in \mathcal{C},~\forall i\in\Zp.
    \end{aligned}
\end{equation}
From the above equation we conclude that $x(iT) \in \mathcal{X}_T$ for all $i\in\Zp$. Note that since $z^-(iT)\in \mathcal X_T\ominus\mathcal C$, $\mathcal C\subset \mathcal D$, and $\mathcal X_T = \mathcal X\ominus\mathcal D$, it follows that $\mathcal D_i = \{z^-(iT)\}\oplus \mathcal D\subset \mathcal X$ for all $i\in \mathbb N$. From \cite[Theorem 1]{garg2021FxTSDomain}, the set $\mathcal D_i$ is forward-invariant for the closed-loop trajectories $x(t)$, i.e., $x(t)\in \mathcal D_i$ for $t\in [(i-1)T, iT)$ for all $i\in \mathbb N$. Hence, it follows that $x(t)\in \mathcal X$ for all $t\geq 0$. Thus, the closed-loop trajectories under the control policy \eqref{eq:policy} satisfy \eqref{eq:lowLevelCnstr}, i.e., the set $\mathcal X_T$ is periodically safe w.r.t. the set $\mathcal X$ with period $T$.
\end{proof}

\section{Simulations}
In this section, we present a simulation case study where we use the proposed strategy to steer a Segway to the origin\footnote{\scriptsize{Code available at
\href{https://github.com/kunalgarg42/fxts_multi_rate}{\texttt{https://github.com/kunalgarg42/fxts\_multi\_rate}}
}}.
The state of the system are the position $p$, the velocity $v$, the rod angle $\theta$ and the angular velocity $\omega$ (see Figure \ref{fig:3 plots}). The control action is the voltage commanded to the motor and the equations of motion used to simulate the system can be found in~\cite[Section~IV.B]{gurriet2018towards}. 
% The nominal model is obtained using a small angle approximation and the MPC is implemented for $Q = \text{diag}(10,10^{-3},10^{-5},10^{-5})$, $R =10^{-3}1$, $Q_F = \text{diag}(10^3,10^3,10^3,2 10^3)$, $\mathcal{U} = \{v ~|~ ||v||_\infty \leq 25\},$ $ \mathcal U_M = \{v ~ | ~ ||v||_\infty \leq 15\}$ and $K = [0,	-7.3989,	-10.435,	-3.7039]$. We choose the set $\mathcal X_T = \{x = [p_x, v_x, \theta, \omega]^T\; |\; |p_x|\leq 10, |v_x|\leq 5, |\theta|\leq 0.3, |\omega|\leq 10\pi\}, \mathcal X_F = \{0\}$, and the parameters $d = 0.2$ and $c = 0.005$. 
% Finally, we implemented the QP \eqref{QP gen} for $\mathcal{S}_e = \{e \in \mathbb{R}^n: e^\top Q_ee \leq 1 \}$ with $Q_e = \text{diag}(1/0.2^2, 1/0.1^2,1/0.05^2,1/0.01^2)$, $\mathcal{X}_c = \mathbb{R}^n$ and $\mathcal{U} = \mathbb{R}$.
In this simulation, we run the high-level MPC planner at $5$Hz and the low-level controller at $10$kHz with parameters $d = 0.6$ and $c = 0.005$. We choose the set $\mathcal X_T = \{x = [p, v, \theta, \omega]^T\; |\; |p|\leq 10, |v|\leq 5, |\theta|\leq 0.3, |\omega|\leq 10\pi\}, \mathcal X_F = \{0\}$, input bounds $\|u\|\leq 25$ with $\|u_m\|\leq 15$. In the first scenario, the initial conditions are $[-1.0 \; 0 \; 0.1 \; 0.3]^T$ and $T = 0.2$. Figure \ref{fig:3 plots} shows the evolution of the FxT barrier functions $h_i$ and the control input $u$. It can be seen that the input constraints are always satisfied, and the FxT barrier functions reach zero at each time step $T$, leading to periodic safety of the underlying set $\mathcal X$. 

\begin{figure}[b]
    \centering
        \includegraphics[width=1\columnwidth,clip]{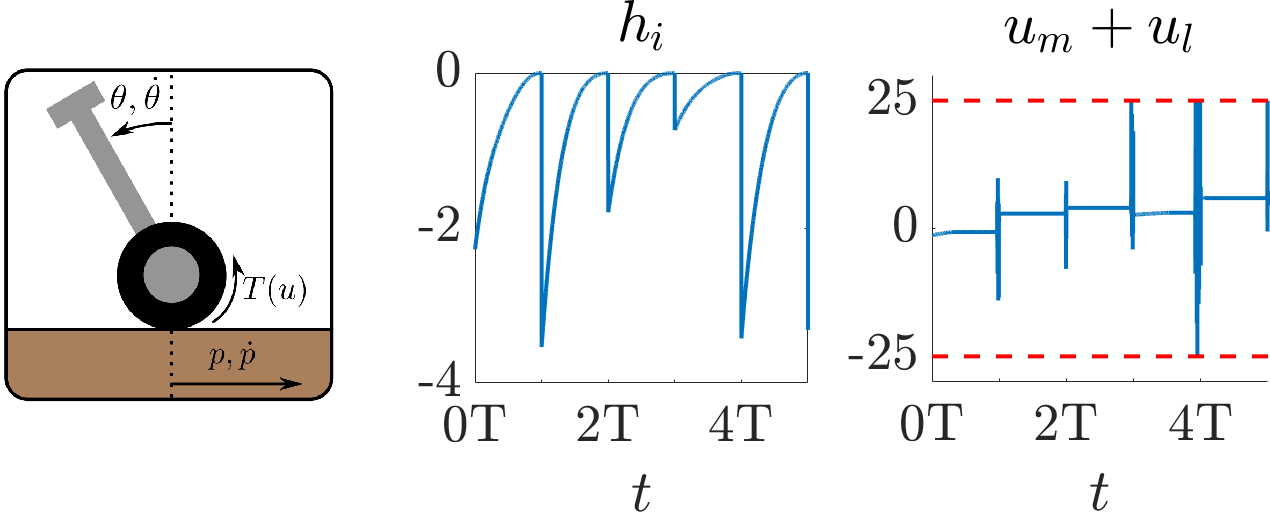}
    \caption{Plots that depict the system and the values of the barrier function $h_i$ and total input $u_m+u_l$.}\label{fig:3 plots}
\end{figure}

\begin{figure*}[!ht]
\centering
\includegraphics[width=\linewidth]{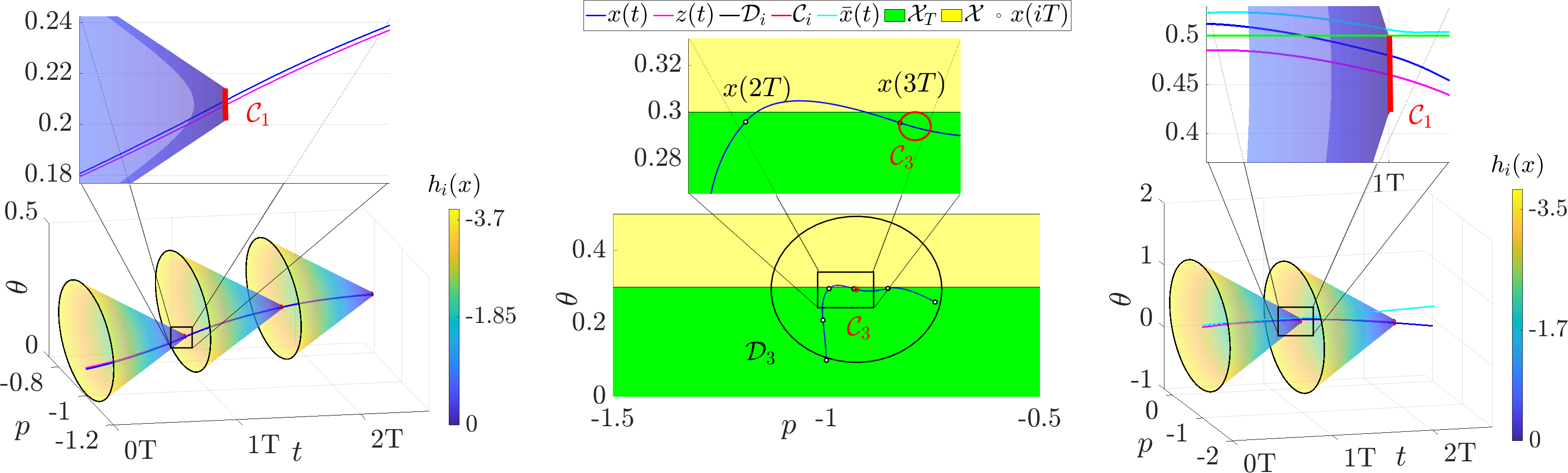}
\caption{Simulation results demonstrating the proposed method. \textbf{(Left)} The trajectory $x(t)$ projected onto $p-\theta-t$ axes. In each interval $[iT, (i+1)T)$, FxT-DoA $\mathcal{D}_i$ is depicted as the colored region which decays to the set $\mathcal C_i$ by the end of the interval. \textbf{(Middle)} Projection of the closed-loop trajectory on the $p-\theta$ plane. The trajectory leaves and enters the set $\mathcal X_T$ after $t = 2T$ and before $t = 3T$, respectively. \textbf{(Right)} Illustration of a scenario where the trajectory $\bar x(t)$ generated using an exponentially stabilizing CLF fails to enter $\mathcal{C}_1$, leading to infeasibility of the MPC at $t = 1T$, whereas the trajectory $x(t)$ generated by the proposed method enters $\mathcal{C}_1$ before $t = 1T$.}
\end{figure*}

% With  $r^\star<\bar r$, $\mathcal D_i$ is a FxT-DoA for the set $\mathcal C_i$ as shown in the left plot of Figure 3, 
Periodic safety is also evident from the left figure in Figure 3, where the closed-loop trajectories are shown to converge to the set $\mathcal C_i$ by end of each interval $\mathcal T_i$. The middle plot in Figure 3 shows the projection of the closed-loop trajectories on the $p-\theta$ plane. It can be seen from the inset plot that the closed-loop trajectory leaves the set $\mathcal X_T$ in the interval $\mathcal T_3$ and returns to the set before the next time step. It can be observed (as discussed in Section \ref{sec:multi-rate strategy}) that the sets $\mathcal C_i, \mathcal D_i$ satisfy $\mathcal C_i\subset\mathcal X_T$ and $\mathcal D_i\subset\mathcal X$, respectively, guaranteeing $x(iT)\in \mathcal X_T$ and $x(t)\in \mathcal X$, i.e., periodic safety of the set $\mathcal X_T$ w.r.t. the set $\mathcal X$.

To compare the performance of the fixed-time stabilizing controller with an exponentially stabilizing one at the low level, we performed a simulation with initial conditions very close to the boundary of the set $\mathcal X_T$. In this case, we chose $\mathcal X_T = \{x = [p, v, \theta, \omega]^T\; |\; |p|\leq 10, |v|\leq 5, |\theta|\leq 0.5, |\omega|\leq 10\pi\}$ and initialized the system with $\theta(0) = 0.495$. In this case, the parameters are chosen as $d = 1, c = 0.04$ and $T = 0.25$. The right plot on Figure 3 shows the trajectory $x(t)$ generated by the proposed controller, and the trajectory $\bar x(t)$ generated by an exponentially stabilizing controller \cite{ames2017control}. The inset plot on the right plot of Figure 3 shows that both $x(t)$ and $\bar x(t)$ leave the set $\mathcal X_T$. The closed-loop trajectory $x(t)$ returns back to the set $\mathcal X_T$ before $t = 1T$, while $\bar x(t)$ fails to do so, leading to infeasibility of the MPC in \eqref{eq:ftocp} at $t = 1T$. This demonstrates the efficacy of the proposed framework over the existing methods that use exponentially stabilizing controllers. 

\section{Conclusions}
In this paper, we introduced the notion of periodic safety requiring system trajectories to visit a subset of a safe set periodically. We defined the notion of fixed-time barrier function and used it in a multi-rate control framework, with MPC as a high-level planner, for control synthesis. We demonstrated that the proposed framework is capable of solving corner cases where exponentially stabilizing controllers might fail. Future work includes studying the robustness properties of the proposed framework by considering model uncertainties.

\bibliographystyle{IEEEtran}
\bibliography{taylor_main,myreferences}

\end{document}